\documentclass[12pt]{article}
\usepackage{latexsym}
\usepackage{amsmath}
\usepackage{amsfonts}
\usepackage{amssymb}
\usepackage{amsthm}
\usepackage{setspace}
\usepackage{graphicx}
\usepackage{tikz}
\usetikzlibrary{arrows}
\usepackage{hyperref}
 \hypersetup{
    colorlinks   = true,
     citecolor    = blue,
     linkcolor=blue,
     urlcolor=blue
}
\usepackage{natbib}
\newcommand\beq{\begin{equation}}
\newcommand\eeq{\end{equation}}
\newcommand\bit{\begin{itemize}}
\newcommand\eit{\end{itemize}}
\newcommand\bea{\begin{eqnarray}}
\newcommand\eea{\end{eqnarray}}
\newcommand\beas{\begin{eqnarray*}}
\newcommand\eeas{\end{eqnarray*}}
\newcommand\beqa{\begin{eqnarray}}
\newcommand\eeqa{\end{eqnarray}}

\newcommand\pfof[1]{{\bf Proof of #1:  }}
\newcommand\eop{\hfill $\square$}

\newtheorem{theorem}{Theorem}
\newtheorem{lemma}{Lemma}
\newtheorem{proposition}{Proposition}

\theoremstyle{remark}
\newtheorem{example}{Example}
\renewcommand\({\left(}
\newcommand\vmax{\overline{v}}
\newcommand\vmin{\underline{v}}
\begin{document}

\title{Distributionally Robust Pricing In Independent Private Value Auctions\footnote{This paper is based on Chapter II of my PhD thesis at Stanford University and comprises a much revised version of an earlier working paper \cite{suzdal2018}. I would like to thank (in random order) Michael Ostrovsky, Andy Skrzypacz, Dmitry Arkhangelsky, Jeremy Bulow, Robert Wilson, Gabriel Carroll, Ilya Segal, Evgeny Drynkin, and audience members at 2017 Conference on Economic Design, York, UK, for helpful comments. } } 
\author{Alex Suzdaltsev\footnote{Higher School of Economics, Saint Petersburg, Russia, \texttt{asuzdaltsev@gmail.com}} 
}
\date{\today}

\maketitle

\begin{abstract}
A seller chooses a reserve price in a second-price auction to maximize worst-case expected revenue when she knows only the mean of value distribution and an upper bound on either values themselves or variance. Values are private and iid. Using an indirect technique, we prove that it is always optimal to set the reserve price to the seller's own valuation. However, the maxmin reserve price may not be unique. If the number of bidders is sufficiently high, all prices below the seller's valuation, including zero, are also optimal. 
A second-price auction with the reserve equal to seller's value (or zero) is an asymptotically optimal mechanism (among all ex post individually rational mechanisms) as the number of bidders grows without bound.  

\end{abstract}

Keywords: Robust mechanism design; Worst-case objective; Auctions; Moments problems

JEL codes: D44, D82

\section{Introduction}
Classic auction theory derives revenue-maximizing reserve prices under the assumption of known distribution of bidders' values. One may give two interpretations to this assumption: (i) the probability distribution is objectively known to the seller; (ii) the distribution represents her beliefs.  Under any of these interpretations, the usual assumption can be problematic: (i) may not hold in practice, especially in the case of new goods; and (ii) contradicts empirical evidence showing that humans do not act as Bayesian decision-makers \citep{ellsberg1961risk}. 

In this paper, we propose an analysis of a simple textbook auction environment with one tweak: assume that the seller, instead of knowing the distribution of bidders' values, knows less, and evaluates the residual uncertainty over the distributions using worst-case criterion. What would be the optimal reserve price?

We assume that bidders' private valuations are known to be independent draws from some unknown distribution $F$. We then consider two specifications of the seller's information: (1) the seller knows the mean of $F$ and an upper bound on values; (2) she knows the mean of $F$ and an upper bound on its variance. One may justify this approach in various ways:
\begin{itemize}
\item It may be easier for the seller to make an educated guess about two numbers than about a whole distribution.
\item It may be easier for the seller to estimate statistically a small number of parameters than a whole distribution. In particular, nonparametric estimators of density functions converge more slowly than parametric estimators of the distribution's moments.
\item As shown by \cite{wolitzky2016mechanism}, a model of a seller who knows only the mean of value distribution and bounds on its support can arise from the seller's uncertainty about bidders' information structures\footnote{In particular, bidder's posterior mean can follow any distribution $F$ with mean $m$ and support in $[0,\vmax]$ if the prior value distribution is a binary distribution on $\{0,\vmax\}$ with mean $m$. To apply this in our setting where values are known to be iid, we must assume that bidders information structures are identical.}. 
\end{itemize}

The question of which reserve price is maxmin may be interesting not only from a normative, but from a positive perspective as well. Empirical literature on reserve prices in auctions\footnote{See, e.g., \cite{mcafee1992updating}; \cite{paarsch1997deriving}; \cite{mcafee2002set}; \cite{haile2003inference}.} yields what \cite{ostrovsky2016reserve} call a ``reserve price puzzle'': the reserve prices observed in auctions are typically substantially lower than reserve prices optimal under the estimated distributions of values. The very formulation of the puzzle suggests that its possible explanation involves postulating that the sellers do not possess estimates produced by econometric studies. This lack of distributional information, coupled with a worst-case perspective employed by a seller could intuitively explain the fact the observed reserve prices are low. Indeed, as a function of the reserve price, the expected revenue under a fixed distribution is typically relatively flat to the left of the optimum and declines sharply to the right of the optimum, so the losses when overshooting the unknown optimal reserve are substantial while the losses when undershooting are minor (see Figure~ \ref{fig:Bayesian}). This suggests that a cautious seller may want to employ a low reserve. 

\begin{figure}[h!]
\centering
\begin{tikzpicture}[>=triangle 45, xscale=6,yscale=4]
\draw[->] (0,0) -- (1.1,0) node[below] {$r$};
\draw[->] (0,0) -- (0,0.8) node[left] {${R}$};

\draw[thick,domain=0:1] 
plot (\x,{pow(\x,3)-3/2*pow(\x,4)+1/2});
\draw[dashed] (0.5,0)--(0.5,17/32);
\node[below] at (0.5,0){$r^*$};

\end{tikzpicture}
\caption{A typical plot of expected revenue as a function of reserve price. The revenue is close to optimal to the left of the optimum and is far from optimal to the right of it. (The plot shown is for $v_i\sim U[0,1]$ and $n=3$.)}\label{fig:Bayesian}
\end{figure}
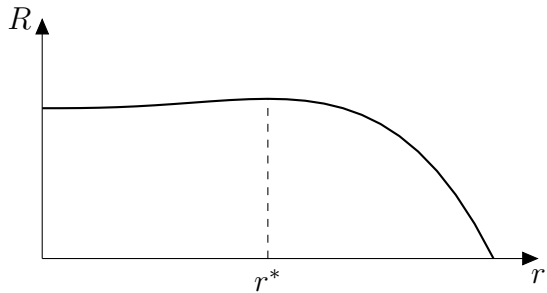

The main results of this paper state exactly this: a seller maximizing the worst-case expected revenue cannot do better than to set the reserve price to her own valuation, in contrast to the classical Bayesian setting where the optimal reserve is higher than the seller's valuation. This is true in both settings we consider. The ``low reserves under ambiguity'' phenomenon is an emerging theme in the literature (see section \ref{litreview} for details); our setting is, to the best of our knowledge, the first one in which the optimal reserve price is low for any fixed number of bidders $n\geq 2$.

However, the optimal reserve price in our setting may be non-unique. Suppose $c$ is the seller's valuation. When the number of bidders is sufficiently high, or $c$ is sufficiently low, there is a whole interval of optimal prices that includes all prices in $[0,c]$ (and possibly some higher prices). 
This is also compatible with empirical evidence: the reserve prices observed in practice are frequently not just low, but are lower than all plausible seller's valuations \citep{hasker2010ebay}. Explanations for this fact proposed in the literature include boundedly rational bidders who do not fully understand how participation rate depends on the reserve price while sorting among competing sellers \citep{jehiel2015absolute} and a combination of value interdependence and bidders' risk aversion \citep{hu2019low}. Even though all reserve prices below $c$ are weakly dominated by $c$ in our model (i.e., they yield a weakly lower revenue than $c$ for any fixed distribution of values), their worst-case optimality provides an additional (weak) explanation for why such prices may be observed in practice. 

Solving the maxmin problems posed in this paper directly by backward induction is, to the best of our knowledge, challenging. So we employ an indirect proof technique. In the first step, using a Lagrangian approach we identify a worst-case distribution $F^*$ if the reserve price is equal to seller's valuation, $r=c$ (this is easier than for other prices). $F^*$ is a binary distribution in the first setting (mean and upper bound on values are known) but has a continuous part in the second setting (mean and upper bound on variance are known). We use an analogy of Nature's problem to a textbook profit maximization problem by a competitive firm with nonconvex costs to minimize the Lagrangian pointwise. Denote the worst-case revenue under $r=c$ by $R^*$. In the second step, for each price $r\geq 0$ we identify a specific distribution $\hat{F}_r$ such that if the reserve price is $r$ and the distribution is $\hat{F}_r$, the seller's revenue is not more than $R^*$. This implies that $r=c$ is a maxmin price. Note that the distributions $\hat{F}_r$ need not be worst-case for respective prices. They are only bad enough to discourage the seller from choosing $r\neq c$, but not necessarily the worst. We find appropriate to call them \emph{threat distributions} -- those Nature may threaten to use to harm the seller if she deviates from $r=c$. This proof technique is similar to one used by \cite{he2020correlation}, who find a maxmin reserve price when a marginal distribution of values, but not their joint distribution, is known (but, of course, the construction of $\hat{F}_r$ is substantially different, as the set of possible distributions is very different from that in \cite{he2020correlation}). 

Sometimes (when the number of bidders is small or $c$ is high) Nature's threats are strong enough so that $r=c$ is the unique maxmin price; when the number of bidders is larger or $c$ is small, other prices, including zero (as noted above), may be maxmin. This indifference occurs because in this case the lowest point in support of $F^*$ happens to be strictly higher than $c$ and the worst-case distribution is still $F^*$ for all $r\leq c$. That is, it may be so that under the worst-case distribution sale always happens, and always happens at a price higher than both the seller's valuation and the reserve price, even when both are positive. We find this somewhat surprising.

In this paper, we address the question of optimal reserve price, but not a more general question of optimal mechanism for a fixed $n$. A technical difficulty that does not allow to use a strong duality approach, as in, e.g., \cite{suzdal2020anoptimal} (see section \ref{litreview}), is that the set of joint distributions of values feasible for Nature is not convex due to the independence constraint (and these constraints are nonlinear). However, we show that the second-price auction without a reserve (or a reserve equal to seller's valuation) is an asymptotically maxmin mechanism in a large class as $n$ grows without bound. This follows from the fact that the revenue guarantee of any individually rational mechanism cannot be higher than the known mean of value distribution, but the revenue guarantee of the second-price auction with a maxmin price converges to this mean as $n$ grows. The result does not even require that bidders play a Bayesian equilibrium. We compare the rates of the convergence of the maxmin revenue to the mean across different settings to obtain insights about the relative ``strength'' of Nature depending on the set of distributions available to it. 

The result that a simple auction is asymptotically optimal must be reconciled with results in \cite{segal2003optimal} who shows that when agents' values are independent draws from an unknown distribution, ``bootstrap'' schemes that estimate the distribution from bidder's reports and set individual prices based on estimates derived from other agents' reports, can asymptotically extract the entire full-distributional-information revenue. These schemes differ substantially from a classic second-price auction. The apparent discrepancy between the our asymptotic result and Segal's is due to the fact that optimization criteria are different: while we employ the maxmin criterion, approximating the full-distributional-information outcome is about minimizing \emph{regret}. When Nature chooses a distribution to minimize revenue itself rather than to maximize losses relative to full information, the ``bootstrap'' schemes cannot be significantly superior to the simple auction.

\subsection{Related literature}\label{litreview}
 
This paper contributes to the growing literature on robust mechanism design. The closest contributions to ours are \cite{carrasco2018optimal}, \cite{he2020correlation}, \cite{koccyiugit2020distributionally}, \cite{suzdal2020anoptimal}, \cite{che2019robust} and \cite{neeman2003effectiveness}. \cite{carrasco2018optimal} study the problem of selling the good to a \emph{single} agent by a seller who maximizes worst-case expected revenue while knowing the first $N$ moments of distribution. They characterize the optimal randomized mechanism; also, they find the optimal deterministic posted price for the single agent for settings (1) and (2) of the present paper.\footnote{The optimal deterministic posted price for the case when the seller knows mean and variance of the value distribution has been found earlier by \cite{azar2013optimal}. In that paper, there is an infinite supply of the good which essentially reduces the environment to a single-agent one.}   \cite{he2020correlation} characterize the optimal deterministic in a second-price auction when the seller knows the marginal distribution of values but not their joint distribution. This setting may be seen as complementary to ours, as we assume that the marginal distribution is unknown, but a particular correlation structure (independence) is known. \cite{he2020correlation} also show that a second-price auction with no reserve is asymptotically optimal among all mechanisms, as in the present paper. Their proof technique is partially similar to ours. 

\cite{koccyiugit2020distributionally} find, among other results, the optimal deterministic reserve price when the seller knows a lower bound for the mean of values, an upper bound for values and there is no restriction on values' correlation structure. \cite{suzdal2020anoptimal} uses strong duality to find an optimal deterministic mechanism for a similar set of distributions where the known means can be heterogeneous. This mechanism happens to be a linear version of the Myersonian optimal auction. \cite{che2019robust} finds an optimal randomized reserve price in a second-price auction for the same set of distributions.
An early precursor to this literature, \cite{neeman2003effectiveness} finds the optimal reserve price in a second-price auction where the set of distributions is the same as in \cite{koccyiugit2020distributionally}, but the criterion is the worst-case ratio of expected revenue to expected full surplus, rather than expected revenue itself.  

Some of the above papers find that with sufficiently many bidders, the robustly optimal reserve price is low. 
In \cite{he2020correlation} and \cite{che2019robust}, the optimal reserve price  converges to seller's value as number of bidders goes to infinity; in \cite{koccyiugit2020distributionally} and \cite{suzdal2020anoptimal}, the optimal reserve is equal to seller's value starting from a certain number of bidders. In contrast, in the present paper it is equal to seller's value for all $n\geq 2$. In this sense, the present paper's setting yields the most striking result among the existing ones. 

Other papers seeking robustness to (payoff) type distributions include  \cite{carrasco2018robust}, \cite{auster2018robust}, \cite{bergemann2011robust}, \cite{bergemann2008pricing}. 
\cite{carroll2017robustness}, \cite{giannakopoulos2019robust} and \cite{chen2019distribution} tackle the problem of selling multiple goods to a single agent under unknown type distribution. \cite{bose2006optimal} and \cite{wolitzky2016mechanism} study mechanism design when agents themselves are maxmin with respect to the distribution of other agents' types.  \cite{wolitzky2016mechanism} uses a specification of sets of possible distributions similar to ours: bounds on support and the mean are known. He gives a microfoundation for this specification which we mentioned earlier. 

A separate strand of literature studies mechanisms robust to misspecification of agents' information structures, rather than the designer's prior. \cite{brooks2019optimal} identify an optimal mechanism in the common value setting, while \cite{du2018robust} identifies a simpler mechanism that asymptotically extracts full surplus. \cite{bergemann2017first} find the optimal robust reserve price in a first-price auction under possible misspecification of agents' information structures.
Relatedly, \cite{chung2007foundations} and \cite{chen2018revisiting} consider robustness to type distributions in a ``rich'' type spaces and identify conditions under which there exist maxmin foundations for dominant-strategy mechanisms. 

Others kinds of robustness explored in the literature include robustness to technology or preferences, robustness to strategic behavior and robustness to interaction among agents and are surveyed by \cite{carroll2018robustness}.

Finally, this paper is related to the literature seeking to explain low reserve prices observed in real-life auctions. \cite{levin1994equilibrium} show that it may be optimal to use a reserve price equal to seller's valuation under endogenous costly entry while \cite{levin1996optimal} show that unlike the textbook IPV case, the optimal reserve converges to the seller's valuation when values are private but are only conditionally iid. As mentioned above, \cite{jehiel2015absolute} and \cite{hu2019low} give explanations for why observed reserve prices are sometimes lower than the seller's valuation. 

\subsection{Organization of the paper} 
In section \ref{set-up}, we describe the set-up. In section \ref{knownmean}, we state and prove the results for the case of known mean and an upper bound on values; in section \ref{knownmeanandvariance} we do the same for the case of known mean and an upper bound on variance. In section \ref{largenumber}, we show that second-price auction without a reserve is an asymptotically optimal mechanism among all mechanisms and compare rates of convergence of the maxmin revenue to its asymptotic value (which is simply the mean of value distribution) for different settings. Section \ref{disc} concludes.

\section{The model}\label{set-up}
Consider the standard second-price auction with one object for sale and $n\geq 2$ bidders (for an extension to first-price auctions, see section \ref{disc}). The valuations of the bidders are iid with some distribution $F$, but $F$ is not fully known to the seller. We consider two specifications of seller's information. In the first one, the seller knows that $\mathbb{E}(v_i)=m>0$ and that $v_i \in[0, \vmax]$, where $\vmax>m$. In the second specification, the seller knows that  $v_i\geq 0$, $\mathbb{E}(v_i)=m>0$ and that $Var(v_i)\leq \sigma^2$.
The variance constraint is specified as inequality, rather than equality due to reasons discussed in \cite{carrasco2018optimal} -- with an equality constraint for the highest moment, the set of distributions may not be compact; also, the proof is somewhat easier to state. However, when variance is known exactly, the results are the same (see section \ref{disc}).

No further restrictions on $F$ are made. In particular, atoms in $F$ are allowed and $F$ is not necessarily regular in Myerson sense.
Denote the set of feasible distributions if only mean and upper bound on values is known by $\Delta_1(m,\vmax)$ and if both mean and upper bound on variance are known by $\Delta_2(m,\sigma^2)$.

The seller's own valuation for the object, $c$, may be higher than bidders' values. We assume that $c\in[0,m)$. One reason for not normalizing $c$ to zero is that we would like to distinguish between a reserve price equal to $c$ and zero reserve price. Also, we would like to allow distributions that put mass below $c$.

The seller wishes to set a deterministic, public reserve price that maximizes revenue. Denote the expected revenue (including the seller's valuation $c$) if the distribution of values is $F$ and the reserve price is $r$ by $R(F,r)$. %
We consider the following problem:
\beq\label{problem}
\sup\limits_{r\geq 0}\inf\limits_{F\in \Delta} R(F,r),
\eeq
where either $\Delta=\Delta_1(\vmax,m)$ or $\Delta=\Delta_2(m,\sigma^2)$. In other words, the seller wishes to set a price in such a way that the worst-case guarantee of revenue given her information is maximal. Denote by $\underline{R}(r)$ the value of the infimum in \eqref{problem}, i.e. the value of this guarantee.  We call any price $r$ that solves \eqref{problem} a \emph{maxmin reserve price} and dthe corresponding expected revenue $R^*$ the \emph{maxmin revenue}.

To present the analysis of the above problem, it will be convenient to us to  phrase it as a zero-sum Stackelberg game between the seller and adversarial Nature in which the seller moves first by setting a price $r$ and then Nature, upon seeing $r$, chooses a distribution $F$ from the choice set $\Delta$.

Denote by $F(v)$ the cdf of the distribution $F$ and by $v_{(i)}$ the $i$th-highest component of the vector of valuations $v$. Then, assuming the bidders play dominant strategies, the function $R(F,r)$ is given by
\begin{multline}\label{revenue}
R(F,r)=c\cdot P(v_{(1)}\leq r)+ r\cdot P(v_{(1)}>r\cap v_{(2)}\leq r)+\mathbb{E}_{F\sim F\sim\cdots\sim F}\left[v_{(2)}\cdot\mathbf{1}_{\{v_{(2)}>r\}}\right]=\\
r-(r-c)F^n(r)+\int_r^{\infty}\(1-nF^{n-1}(v)+(n-1)F^n(v)\right)dv,
\end{multline}
where we used an expression relating the cdf of second-order statistic to the cdf of the parent distribution $F$ and the identity $\mathbb{E}(X)=\int_0^{+\infty}(1-F(v))dv$ for a nonnegative random variable $X$ with cdf $F$. 

Expressing the expected revenue in terms of the cdf $F(v)$ allows to simultaneously cover all distributions regardless of presence of atoms, and also allows to reduce optimization over distributions to optimization over functions.

\section{Known mean and upper bound on values}\label{knownmean}
\subsection{The result}
Consider the problem \eqref{problem} with $\Delta=\Delta(m,\vmax)$. Among other results, \cite{carrasco2018optimal} solve this problem for $n=1$, i.e., solve the monopolistic pricing problem. They show that there exists a unique maxmin price that exceeds seller's costs.

In contrast, a main result of this paper is when there are at least two bidders, a reserve price equal to seller's opportunity costs $c$ is maxmin.

\begin{theorem}[\textbf{Main result I}]\label{main}
Suppose the seller knows the mean of value distribution $m$ and an upper bound on values $\vmax$. Then, the set of prices $r^*$ solving problem \eqref{problem} includes the seller's valuation $c$.
\end{theorem}

\subsection{The proof} 
In this section, we provide the proof of theorem \ref{main}.
The plan of attack, as outlined in the introduction, consists of two steps. In the first step, we identify the worst-case distribution $F^*$ when $r=c$ using a Lagrangian method. In the second step, for any $r\geq 0$ we specify a distribution $\hat{F}_r$ such that $R(\hat{F}_r,r)\leq R(F^*,c)$, without making any claim that $\hat{F}_r$ is worst-case. This approach allows to circumvent the need to solve for a worst-case distribution for each $r$. Identifying a worst-case distribution for $r=c$ is much simpler than for other prices because the term $-(r-c)F^n(r)$ in \eqref{revenue} disappears when $r=c$, so that the expected revenue depends on $F$ only though an integral. In the previous version of this paper, \citep{suzdal2018}, we do identify worst-case distributions for each $r\geq 0$ but this requires, at a point, tedious second-order analysis and works only for the case of a bound on values, but not a bound on variance. 

\subsubsection{First step}
Suppose $r=c$. Then, Nature's problem, as a problem of choosing a function $F(v)$, may be written as
\begin{align}
\min\limits_{F(\cdot)} & \left(c + \int_c^{\vmax}\(1-nF^{n-1}(v)+(n-1)F^n(v)\right)dv\right) \label{controlproblem}\\
\mbox{s.t. } & \int_0^{\vmax}(1-F(v))dv=m \label{mean}\\
& F(v)\in [0,1] \mbox{ for all } v\in[0,\vmax] \label{bounds}\\
& F(v) \mbox{ is nondecreasing}\label{incr} \\
& F(v) \mbox{ is right-continuous}\label{cont} 
\end{align}
The constraint \eqref{mean} is the mean constraint, while the constraints  
\eqref{bounds}-\eqref{cont} are necessary and sufficient to ensure that the function $F(\cdot)$ chosen by Nature is a cdf. The constraint \eqref{cont} is not an issue; the monotonicity constraint may a priori be an issue, but, fortunately, turns out not to be.

To solve the problem \eqref{controlproblem}-\eqref{cont}, we first prove that it is without loss of generality to look at distributions putting no mass below $c$. This allows to make the integration bounds in the objective \eqref{controlproblem} and the constraint \eqref{mean} the same.

\begin{lemma}\label{nomassbelowc}
For every feasible cdf $F$ in problem \eqref{controlproblem}-\eqref{cont}, there exists a feasible cdf $\tilde{F}$ putting no mass below $c$ such that the expected revenue \eqref{controlproblem} is weakly lower under $\tilde{F}$ than under $F$.
\end{lemma}

\pfof{lemma \ref{nomassbelowc}}{
Take a feasible cdf $F$. Define $\beta:=(m-c)/\int_c^{\infty}(1-F(v))dv$. Because $F$ is feasible, $\beta\in(0,1]$. Then consider $$\tilde{F}(v):=
\begin{cases}
0, & v<c \\
\beta F(v)+(1-\beta), & v\geq c.
\end{cases}$$
By construction, $\int_0^{\infty}(1-\tilde{F}(v))dv=m$ so $\tilde{F}$ is feasible and puts no mass below $c$. The revenue is weakly lower under $\tilde{F}$ than under $F$ because $\tilde{F}(v)\geq F(v)$ for all $v\geq c$ and the integrand in \eqref{controlproblem} is decreasing in $F$.
\eop}

The intuition behind lemma \ref{nomassbelowc} is straightforward. Suppose there is some probability mass strictly below $r=c$. By transferring it all to $r=c$ Nature will not change the expected revenue, but will increase the mean of the distribution. Then it can restore the mean by redistributing mass within the set $\{v: v\geq c\}$ towards lower values which will reduce the revenue. 

In light of lemma \ref{nomassbelowc}, the mean constraint can be now rewritten as 
\beq\label{meannew}
\int_c^{\vmax}(1-F(v))dv=m-c
\eeq

This allows to proceed to forming a Lagrangian.

Define the Lagrangian by
\beq\label{lagrangian}
\mathcal{L}(F,\lambda):=\int_c^{\vmax}\left(1-n F^{n-1}(v)+(n-1)F^n(v)+\lambda(1- F(v))\right)dv
\eeq
In what follows, we will minimize the Lagrangian pointwise. 
The validity of the Lagrangian approach rests on the following lemma:
\begin{lemma}\label{sufficiency0}
Suppose $F_0$ is a cdf that minimizes the Lagrangian among all cdfs for some $\lambda\in\mathbb{R}$ and satisfies \eqref{meannew}. Then, $F_0$ solves the problem \eqref{controlproblem}.
\end{lemma}

\begin{proof}
Take $F_0$ and any other cdf $\tilde{F}$ satisfying \eqref{meannew} (which is without loss of generality by lemma \ref{nomassbelowc}). Because $F_0$ minimizes the Lagrangian, we have 
\[\mathcal{L}(F_0,\lambda)\leq \mathcal{L}(\tilde{F},\lambda).\]
Because both $F_0$ and $\tilde{F}$ satisfy \eqref{meannew},
\[-\lambda\int_c^{\vmax}(1-F_0(v))dv=-\lambda\int_c^{\vmax}(1-\tilde{F}(v))dv.\]
Summing the above relations, one gets that $R(F_0,c)\leq R(\tilde{F},c)$, as desired.
\eop 
\end{proof}

Call the integrand in \eqref{lagrangian} $H(F,\lambda)$\footnote{The notation stems from the fact that the integrand is equal to the Hamiltonian of the corresponding optimal control problem. The Minimum Principle (as applied to the relaxed problem) guarantees the existence of the Lagrange multiplier $\lambda$ such that the optimal $F$ minimizes the Lagrangian pointwise. However, in the formal proof we construct the multiplier explicitly and therefore do not have to rely on Minimum Principle.}. The first-order condition for the minimization of $H$ with respect to $F$ is 
\beq\label{foc}
\lambda=n(n-1)F^{n-2}(F-1).
\eeq 
If $n=2$, this equation has the unique solution, so the optimal $F$ is a constant which corresponds to a binary distribution on $\{c,\vmax\}$. For $n\geq 3$, however, $H(F,\lambda)$ is not convex in $F$. In fact, for $n\geq 3$, Nature's problem is isomorphic to a textbook profit maximization problem of a competitive firm with U-shaped marginal and average costs functions that chooses its ``output" $F$ given the ``price" $\lambda$. (The corresponding ``total cost function'' is $TC(F)=(n-1)F^n-nF^{n-1}$.) If the price is below the minimum of average costs ($\min AC$), the optimal output is zero; if the price is above $\min AC$, the optimal output is given by the minimum of the larger solution to \eqref{foc} and 1; and if the price is exactly equal to $\min AC$, both zero output and the output minimizing the AC are optimal.

Define $q^*_n:=1-\frac{1}{(n-1)^2}$. This is the ``output'' minimizing ``average costs''. Define also 
\[z(y):=y^{n-1}-y^{n-2}.\]
The ``supply curve" stemming from the pointwise minimization of the Lagrangian (maximization of $H$) is stated in the following lemma (we omit its proof):
\begin{lemma}\label{supplylemma}
\beq\label{supply}
\arg\max\limits_{F\in[0,1]}H(F,\lambda)=
\begin{cases}
\{0\},& \lambda<\lambda^*;\\
\left\{0,q^*_n\right\},& \lambda=\lambda^*;\\
\{\min\{\bar{y}(\lambda),1\}\}, & \lambda>\lambda^*,
\end{cases}
\eeq
where $\bar{y}(\lambda)$ is the larger solution to \eqref{foc} and 
\beq\label{lambdastar}
\lambda^*=\min\limits_{y\in[0,1]}\left[\frac{(n-1)y^n-ny^{n-1}}{y}\right]=n(n-1)z(q^*_n).
\eeq
\end{lemma}
$\lambda^*$ is ``$\min AC$" (see Figure~\ref{fig:supplycurve} for the case $n=3$).

\begin{figure}[h]]
\centering  
\begin{tikzpicture}[>=triangle 45, xscale=5,yscale=2.5]
\draw[->] (0,0) -- (1.3,0) node[below] {$F$};
\draw[->] (0,-1.7) -- (0,0.3) node[left] {$\lambda$};
\draw[thick, domain=0:1] plot (\x, {6*\x*(\x-1)});
\node[right] at (0.7,-1.4){``MC''};
\node[right] at (1.05,-1){``AC''};
\draw[red, ultra thick, domain=0.75:1] plot (\x, {6*\x*(\x-1)});
\draw[thick,domain=0:1.3] plot(\x,{2*\x*\x-3*\x});
\draw[red, ultra thick] (0,-1.7) -- (0,-9/8);
\draw[red, ultra thick] (1,0)--(1,0.3);
\draw[dashed] (0,-9/8)--(3/4, -9/8)--(3/4, 0);
\node[left] at (0,-9/8) {$\lambda^*$};
\node[above] at (3/4,0) {$q^*_3=0.75$};
\end{tikzpicture}
\caption{The pointwise minimization of the Lagrangian for $n=3$ (the argmin correspondence is in red). }\label{fig:supplycurve}
\end{figure}
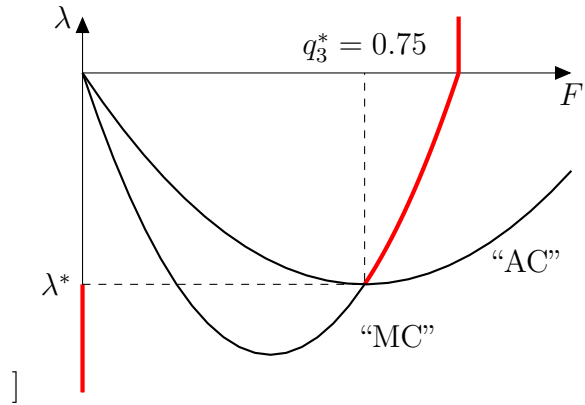

Fortunately, a value of $\lambda$ can always be found such that $F(v)$ maximizes $H(F,\lambda)$ for all $v\in[c,\vmax]$ and the mean constraint is satisfied. This leads to the identification of worst-case distributions. 
Define 
\[\vmin^*:=
\begin{cases}
0, & n=2;\\
\max\left\{m-\frac{\vmax-m}{(n-1)^2-1},0\right\} & n\geq 3.
\end{cases}\]

Denote by $\delta_{a,b}$ a binary distribution with support $\{a,b\}$ and mean $m$. Note that if $\vmin^{*}>0$, the distribution $\delta_{\vmin^{*},\vmax}$ puts a probability of $q^*_n$ on $\vmin^{*}$. 

\begin{proposition}\label{wcmean}
Suppose $r=c$. Then the distribution $\delta_{\max\{\vmin^{*},c\},\vmax}$ solves problem \eqref{controlproblem}-\eqref{cont}.
\end{proposition}

\pfof{proposition \ref{wcmean}}{
Suppose $c>\vmin^*$. This is equivalent to $(n-1)^2<\frac{\vmax-c}{m-c}$. The binary distribution on $\{c,\vmax\}$ is described by a cdf $F^*(v)=\frac{\vmax-m}{\vmax-c}$ for all $v\in[c,\vmax)$. Denote $p=\frac{\vmax-m}{\vmax-c}$.  Consider $\lambda_0 =n(n-1)z(p)$. Because $(n-1)^2<\frac{\vmax-c}{m-c}$, $p>q^*_n$ so $\lambda_0>\lambda^*$. Thus, by lemma \ref{supplylemma}, $F^*$ minimizes the Lagrangian for $\lambda=\lambda_0$. By lemma \ref{sufficiency0}, $F^*$ solves the problem  \eqref{controlproblem}-\eqref{cont}.

Now suppose $c\leq \vmin^*$, i.e. $(n-1)^2\geq \frac{\vmax-c}{m-c}$. The binary distribution on $\{\vmin^*,\vmax\}$ is such that its cdf $F^*$ takes values 0 and $q^*_n$ on $[c,\vmax)$. Consider $\lambda_0 =\lambda^*$. Thus, by lemma \ref{supplylemma}, $F^*$ minimizes the Lagrangian for $\lambda=\lambda_0$. By lemma \ref{sufficiency0}, $F^*$ solves the problem  \eqref{controlproblem}-\eqref{cont}.
\eop}

It is instructive to consider a specific numeric example. Note that if $c=0$, Nature's problem for $r=c$ is one of finding a distribution minimizing the expectation of second-order statistic in a sample given the known mean and upper bound. 

\begin{example}
Suppose $n=3$, $\vmax=1$ and $m=1/2$, and $r=c=0$. Then, under the worst-case distribution the valuation of each bidder is equal to 1/3 with probability 3/4 and is equal to 1 with probability 1/4. The worst-case expected revenue is equal to 7/16.
\end{example}

As noted in the introduction, the fact that the support of the worst-case distribution may be bounded away from $c$ is a reason for why the maxmin reserve price may not be unique. 

\subsubsection{Second step}
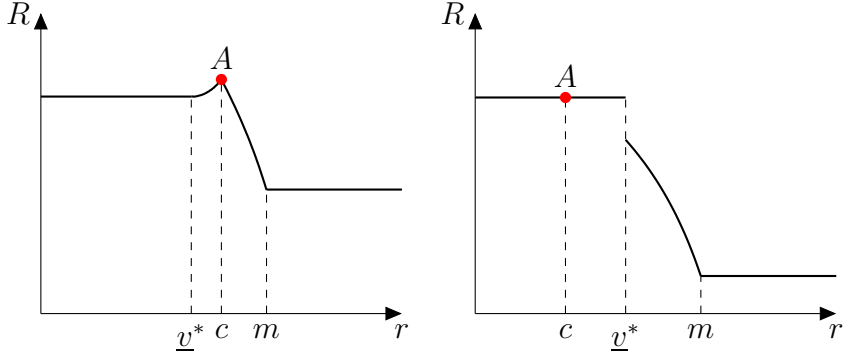
\begin{figure}[h!]
\centering
\begin{tikzpicture}[>=triangle 45, xscale=6,yscale=33]
\draw[->] (0,0) -- (0.8,0) node[below] {$r$};
\draw[->] (0,0) -- (0,0.1212) node[left] {${R}$};

\draw[thick] (0,7/16-0.35)--(0.3333,7/16-0.35);
\draw[thick, domain=0.3333:0.4] 
plot (\x, {1-(1-\x)*pow(0.5/(1-\x),3)-3*(0.5-\x)*pow(0.5/(1-\x),2)-0.35});

\draw[thick, domain=0.4:0.5] 
plot (\x, {1-(1-0.4)*pow(0.5/(1-\x),3)-3*(0.5-\x)*pow(0.5/(1-\x),2)-0.35});
\draw[thick] (0.5,0.05)--(0.8,0.05);
\node[below] at (0.4,0){$c$};
\node[below] at (0.5,0){$m$};
\node[below] at (0.3333,0){$\vmin^*$};
\draw[dashed] (0.5,0)--(0.5,0.05);
\draw[dashed] (0.4,0)--(0.4,0.4444-0.35);
\draw[dashed] (0.3333,0)--(0.3333,7/16-0.35);
\node[circle,fill=red,inner sep=1.5pt,minimum size=0pt] at (0.4,0.4444-0.35) {};
\node[above] at (0.4,0.4444-0.35){$A$};

\end{tikzpicture}
\begin{tikzpicture}[>=triangle 45, xscale=6,yscale=10]
\draw[->] (0,0) -- (0.8,0) node[below] {$r$};
\draw[->] (0,0) -- (0,0.4) node[left] {${R}$};

\draw[thick] (0,7/16-0.15)--(0.3333,7/16-0.15);

\draw[thick, domain=0.3333:0.5] 
plot (\x, {1-(1-0.2)*pow(0.5/(1-\x),3)-3*(0.5-\x)*pow(0.5/(1-\x),2)-0.15});
\draw[thick] (0.5,0.2-0.15)--(0.8,0.2-0.15);
\node[below] at (0.2,0){$c$};
\node[below] at (0.5,0){$m$};
\node[below] at (0.3333,0){$\vmin^*$};
\draw[dashed] (0.5,0)--(0.5,0.2-0.15);
\draw[dashed] (0.2,0)--(0.2,7/16-0.15);
\draw[dashed] (0.3333,0)--(0.3333,7/16-0.15);
\node[circle,fill=red,inner sep=1.5pt,minimum size=0pt] at (0.2,7/16-0.15) {};
\node[above] at (0.2,7/16-0.15){$A$};

\end{tikzpicture}
\caption{Proof idea. The curves are graphs of $R(\hat{F}_r,r)$. The graph of the worst-case revenue function $\underline{R}(r)$ must lie everywhere weakly below the depicted curve by Step 2 and must pass through point $A$ by Step 1. $n=3$; $\vmax=1$, $m=1/2$. In the left picture, $c=0.4>1/3=\vmin^*$; $c$ must be the unique maxmin price. On the right, $c=0.2<1/3=\vmin^*$; $c$ may not be uniquely optimal.}\label{fig:proofidea1}
\end{figure}

In this step, for every $r$ we construct a feasible distribution $\hat{F}_r$ such that $R(\hat{F}_r,r)\leq \underline{R}(c)$ for all $r$. This implies the result of theorem \ref{main}. The distributions $\hat{F}_r$ are not necessarily worst-case given a reserve $r$. They may be thought of as \emph{threat distributions}: distributions that Nature threatens to use were the seller to deviate from $r=c$.  

The construction depends on whether $c>\vmin^*$ or $c\leq \vmin^*$.

\vspace{1em}
\textbf{Case 1.} $c>\vmin^*$. The construction of $\hat{F}_r$ is separate for $r\in[0,\vmin^*)$, $r\in[\vmin^*,c)$, $r\in[c,m)$, $r\geq m$. 
Define threat distributions $\hat{F}_r$ by 
\[\hat{F}_r:=
\begin{cases}
\delta_{\vmin^*,\vmax}, & r\in [0,\vmin^*);\\
\delta_{r^+,\vmax}, & r\in [\vmin^*,c);\\
\delta_{r,\vmax}, & r\in[c,m);\\
\delta_m, & r\geq m,
\end{cases}\]
where $r^+$ is a point arbitrarily close to $r$ to the right of it. (Formally, in this case we consider a sequence of distributions $\hat{F}^k_r$, each of those binary on $\{r+1/k,\vmax\}$.)

\begin{proposition}\label{fhatworks}
Suppose $c>\vmin^*$. Then, $R(\hat{F}_r,r)< \underline{R}(c)$ for all $r\neq c$ and $R(\hat{F}_r,r)=\underline{R}(c)$ for $r=c$.
\end{proposition}

\begin{proof}
The fact that $R(\hat{F}_r,r)=\underline{R}(c)$ for $r=c$ is obvious since $\hat{F}_c$ is a worst-case distribution for $r=c$, as identified by proposition \ref{wcmean}.

Consider $r\in[c,m)$. Denote by $p(r)=\frac{\vmax-m}{\vmax-r}$ the probability assigned by $\hat{F}_r$ to $r$. Then, by \eqref{revenue},
\[R(\hat{F}_r,r)=r-(r-c)p^n(r)+(\vmax-r)(1-np^{n-1}(r)+(n-1)p^n(r)).\]
Thus, the full derivative of revenue is
\[\frac{dR(\hat{F}_r,r)}{dr}=np^{n-1}-np^n-np^{n-1}p'+(\vmax-r)n(n-1)z(p)p'.\]
Since $p'(\vmax-r)=p$, this simplifies to
\[n(n-2)pz(p)-np^{n-1}p'\cdot(r-c).\]
Because $p'>0$ ,$\frac{dR(\hat{F}_r,r)}{dr}<0$ except when $r=c$ and $n=2$. Thus, $R(\hat{F}_r,r)$ is strictly decreasing on $[c,m)$, which implies the result.

For $r\geq m$, $R(\hat{F}_r,r)=c$, but $c=\lim\limits_{r\to m^-}  R(\hat{F}_r,r)$. Thus, $c<\underline{R}(c)$. 

Consider now $r\in[\vmin^*,c)$. Denote by $\delta_{r,\vmax}$ the binary distribution on $\{r,\vmax\}$. $R(\delta_{r^+,\vmax},r)$ differs from $R(\delta_{r,\vmax},r)$ only by term $(r-c)p^n$, because  there is a sale if $v_{(1)}=r^+$ but not if $v_{(1)}=r$. Thus, when $r\in[\vmin^*,c)$, we get
\[\frac{dR(\hat{F}_r,r)}{dr}=n(n-2)pz(p)-np^{n-1}p'\cdot(r-c)+[(r-c)p^n]'=p^{n-1}((n-1)^2p-n(n-2)).\]
For $r> \vmin^*$, $p(r)> q^*_n$, and for $r=\vmin^*$, $p(r)=q^*_n$. Thus, the above derivative is zero at $r=\vmin^*$ and positive at $r\in(\vmin^*,c)$. Thus, $R(\hat{F}_r,r)$ is strictly increasing on $[\vmin^*,c)$, which implies the result.

Finally, because $R(\hat{F}_r,r)<\underline{R}(c)$ for $r=\vmin^*$, this is true for $r<\vmin^*$ as well. \eop
\end{proof}

\vspace{1em}
\textbf{Case 2.} $c\leq \vmin^*$. Define threat distributions $\hat{F}_r$ by
\[\hat{F}_r:=
\begin{cases}
\delta_{\vmin^*,\vmax}, & r\in [0,\vmin^*);\\
\delta_{r,\vmax}, & r\in[\vmin^*,m);\\
\delta_m, & r\geq m,
\end{cases}\]

\begin{proposition}\label{fhatworks2}
Suppose $c\leq\vmin^*$. Then, $R(\hat{F}_r,r)\leq \underline{R}(c)$ for all $r\neq c$ and $R(\hat{F}_r,r)=\underline{R}(c)$ for $r=c$.
\end{proposition}

\begin{proof}
For $r\in[0,\vmin^*)$, $R(\hat{F}_r,r)=\underline{R}(c)$, because $\hat{F}_r=F^*$ and the reserve price does not affect the auction. For $r\in[\vmin^*,m)$, by the same reasoning as above, $R(\hat{F}_r,r)$ is strictly decreasing. Finally, $R(\hat{F}_{\vmin^*},\vmin^*)\leq \underline{R}(c)$ because $R(\hat{F}_{\vmin^*},\vmin^*)$ is the revenue under $F^*$ when the good is not sold when all values are equal to $\vmin^*$ while $\underline{R}(c)$ is the revenue under $F^*$ when the good is sold for $\vmin^*$ when all values are equal to $\vmin^*$. \eop
\end{proof}

Now, because by propositions \ref{fhatworks} and \ref{fhatworks2}, 
$R(\hat{F}_r,r)\leq \underline{R}(c)$ for all $r$ and because distributions $\hat{F}_r$ are feasible to Nature, we conclude that $\underline{R}(r)\leq R(\hat{F}_r,r)\leq \underline{R}(c)$ for all $r$, which finishes the proof of theorem \ref{main}.

\subsection{Uniqueness}\label{uniqueness1}
Theorem \ref{main} establishes that $c$ is a maxmin reserve price. But are there other maxmin prices?

\begin{proposition}\label{unique?}
Suppose the seller knows the mean of value distribution $m$ and an upper bound on values $\vmax$. Then:
\begin{enumerate}
\item If  $\vmin^*<c$ (equivalently, $(n-1)^2<\frac{\vmax-c}{m-c}$), $c$ is the unique maxmin reserve price;
\item If $\vmin^*\geq c$ (equivalently, $(n-1)^2\geq \frac{\vmax-c}{m-c}$), all prices $r\in[0,c]$ are maxmin reserve prices.  
\end{enumerate}
\end{proposition}

\begin{proof}
Part 1 follows directly from proposition \ref{fhatworks}. To prove part 2, we show that $F^*$, the distribution identified by proposition \ref{wcmean}, part 2, is a worst-case distribution not only for $r=c$, but for all $r\in[0,c)$ as well. Indeed, if $r<c$, the expected revenue may be written as 
\[R(F,r)=r-c+(c-r)F^n(r)+R(F,c).\]
$F^*$ minimizes $R(F,c)$, but it also minimizes $(c-r)F^n(r)$ because $c-r>0$ and $F^*(r)=0$, as $r<c\leq \vmin^*$. Thus, it minimizes the sum of these two terms. \eop
\end{proof}
 
As discussed in the introduction, proposition \ref{unique?}, part 2, might weakly explain why sometimes reserve prices lower than seller's valuation are observed in real-life auctions. By using the qualifier ``weakly'' we emphasize the caveat that all prices below $c$ are weakly dominated by $c$ (yield weakly lower revenue for any \emph{fixed} distribution) and thus might be refined away despite being worst-case optimal.  

Note, however, that proposition \ref{unique?}, part 2, does not say that prices $r\in[0,c]$ are the only maxmin prices. Indeed, in the previous version of this paper \citep{suzdal2018}
we show that when the maxmin price is not unique, the set of maxmin prices may also include some prices higher than $c$. As noted above, the full characterization of the set of maxmin prices by backward induction requires subtler analysis that is beyond the scope of this paper. 

\section{Known mean and upper bound on variance}\label{knownmeanandvariance}
\subsection{The result}
In this section, we consider the problem \eqref{problem} for $\Delta=\Delta_2(m,\sigma^2)$, that is, consider a situation in which seller knows the mean and an upper bound for variance of value distribution. 

Again, it has been shown that if $n=1$, there exists a unique maxmin price that exceeds seller's costs  (see \cite{azar2013parametric} and \cite{carrasco2018optimal}). In contrast, we show that if $n\geq 2$, the seller can do no better than to set the reserve price to her own valuation. 

\begin{theorem}[\textbf{Main result II}]\label{main2}
Suppose the seller knows the mean of value distribution $m$ and an upper bound on its variance $\sigma^2$. Then, the set of prices $r^*$ solving problem \eqref{problem} includes the seller's valuation $c$.
\end{theorem}

\subsection{The proof}
The plan of proof is exactly the same as in the section \ref{knownmean}. 

\subsubsection{First step}
Suppose $r=c$. As compared with \eqref{controlproblem}, Nature's problem now involves one more constraint. To write it in an integral form, note that, for a nonnegative random variable $v$ with cdf $F(\cdot)$,  $\mathbb{E}(v^2)=\int_0^{\infty} (1-F(\sqrt{s}))ds=\int_0^{\infty} 2v(1-F(v))dv.$ Hence, the additional constraint is $\int_0^{\infty} 2v(1-F(v))dv\leq m^2+\sigma^2$.

Thus, the new Nature's problem is:
\begin{align}
& \min\limits_{F(\cdot)}\left(c+\int_c^{\infty}\(1-nF^{n-1}(v)+(n-1)F^n(v)\right)dv\right)\label{controlproblem2}\\
\mbox{s.t. } & \int_0^{\infty}(1-F(v))dv=m \label{mean2}\\
& \int_0^{\infty}2v(1-F(v))dv\leq m^2+\sigma^2\label{var}\\
& F(v)\in [0,1] \mbox{ for all } v\in[0,\infty) \label{bounds2}\\
& F(v) \mbox{ is nondecreasing, right-continuous}\label{incr2} \\
& \lim\limits_{v\to \infty}F(v)=1 \label{cont2} 
\end{align}
Together, constraints \eqref{mean2} and \eqref{var} ensure that the mean of $F(\cdot)$ is equal to $m$, and its variance is no more than $\sigma^2$. The constraints \eqref{incr2}, \eqref{cont2} ensure that $F(\cdot)$ is a cdf.

Again, before proceeding to a Lagrangian, we show that Nature can restrict itself to distributions putting no mass below $c$. 

\begin{lemma}\label{nomassbelowc2}
For every feasible cdf $F$ in problem \eqref{controlproblem2}-\eqref{cont2}, there exists a feasible cdf $\tilde{F}$ putting no mass below $c$ such that the expected revenue \eqref{controlproblem2} is weakly lower under $\tilde{F}$ than under $F$.
\end{lemma}

\pfof{lemma \ref{nomassbelowc2}}{
The proof is the same as the proof of lemma \ref{nomassbelowc}. The only difference is that one has to show that $\tilde{F}$, as constructed in the proof of lemma \ref{nomassbelowc}, satisfies the variance constraint \eqref{var}. But this is true because $F$ is a mean-preserving spread of $\tilde{F}$, as $\int_0^v F(v)dv\geq  \int_0^v \tilde{F}(v)dv $ and their means are the same.
\eop} 

Lemma \ref{nomassbelowc2} allows to rewrite mean and variance constraints as \eqref{meannew}
and
\beq\label{varnew}
\int_c^{\infty}2v(1-F(v))dv\leq m^2+\sigma^2-c^2.
\eeq

Define the Lagrangian by 
\beq\label{lagrangian2}
\mathcal{L}(F,\lambda_1,\lambda_2)=\int_c^{\infty}\left(1+\lambda_1+2t\lambda_2-n F^{n-1}(v)+(n-1)F^n(v)-(\lambda_1+2\lambda_2 v) F(v)\right)dv
\eeq

The sufficiency of the pointwise minimization of the Lagrangian is now slightly subtler as now we have an inequality constraint. It can be ensured if $\lambda_2$ has the right sign and a candidate worst-case distribution $F_0$ satisfies \eqref{varnew} with equality.

\begin{lemma}\label{sufficiency}
If $F_0$ is any cdf such that (1) $F_0$ minimizes the Lagrangian among all cdfs for some $\lambda_2\geq 0$, $\lambda_1$ of any sign; (2) $F_0$ satisfies \eqref{meannew} and satisfies \eqref{varnew} with equality, then $F_0$ solves the problem  \eqref{controlproblem2}-\eqref{cont2}. 
\end{lemma}

\begin{proof}
Take any cdf $\tilde{F}$ satisfying constraints \eqref{meannew}-\eqref{varnew}. We shall prove that $R(F_0,c)\leq R(\tilde{F},c)$ for any $F_0$ satisfying conditions in the lemma. Because $F_0$ minimizes the Lagrangian, 
\[\mathcal{L}(F_0,\lambda_1,\lambda_2)\leq \mathcal{L}(\tilde{F},\lambda_1,\lambda_2).\]
Because both $F_0$ and $\tilde{F}$ satisfy \eqref{meannew}, 
\[-\lambda_1\int_c^{\infty}(1-F_0(v))dv=-\lambda_1\int_c^{\infty}(1-\tilde{F}(v))dv.\]
Because $F_0$ satisfies \eqref{varnew} with equality, $\tilde{F}$ satisfies \eqref{varnew}, and $\lambda_2\geq 0$, 
\[-\lambda_2\int_c^{\infty}2v(1-F_0(v))dv=-\lambda_2(m^2+\sigma^2-c^2)\leq-\lambda_2\int_c^{\infty}2v(1-\tilde{F}(v))dv.\]
Summing up the above relations, one gets
\[
\int_c^{\infty}\left(1-n F_0^{n-1}(v)+(n-1)F_0^n(v)\right)dv\leq 
\int_c^{\infty}\left(1-n \tilde{F}^{n-1}(v)+(n-1)\tilde{F}^n(v)\right)dv,\]
or $R(F_0,c)\leq R(\tilde{F},c)$.
\eop
\end{proof}

Recall that $z(y)=y^{n-1}-y^{n-2}$ and define 
\beq\label{phidef}
\phi(q):=\frac{\int_q^1(z(y)-z(q))^2dy}{\left(\int_q^1(z(y)-z(q))dy\right)^2}
\eeq
for $q\in (0,1)$. 

Recall from section \ref{knownmean} that $q^*_n=1-\frac{1}{(n-1)^2}$. Now define 
\beq
\vmin^{**}:=\max\left\{ m-\frac{\sigma}{\sqrt{\phi(q^*_n)-1}},0\right\}.
\eeq

Analogously to section \ref{knownmean}, $\vmin^{**}$ will be shown to be the lowest point in support of the worst-case distribution if $c=r=0$. 

We now introduce the family of distributions that plays a major role in both steps of the proof of theorem \ref{main2}. The shape of the distribution is dictated by the minimization of Lagrangian when $r=c$. Given a parameter $\rho\in[\vmin^{**},m)$ define a cdf $G_{\rho}(\cdot)$ as follows: $G_{\rho}(v)=0$ for $v<\rho$; $G_{\rho}(\rho)\equiv q(\rho)$ and
\beq\label{hatFdef}
G^{-1}_{\rho}(q)=\frac{n(n-1)z(q)-\lambda_1(\rho)}{2\lambda_2(\rho)}
\eeq
for $q\in[q(\rho),1]$ where $\lambda_1(\rho)$, $\lambda_2(\rho)$ are parameters attuned in such a way that mean and variance constraints hold as equalities. 

Equivalently, for every $v\geq \rho$, $G_{\rho}(v)=\min\{\bar{y}(\lambda_1(\rho)+2\lambda_2(\rho)v),1\}$ where $\bar{y}(\lambda)$ is the larger solution to \eqref{foc} (as in \eqref{supply}). 
For $n=2$, $G_{\rho}$ is linear for $v\geq \rho$, so the continuous part of the distribution is uniform.
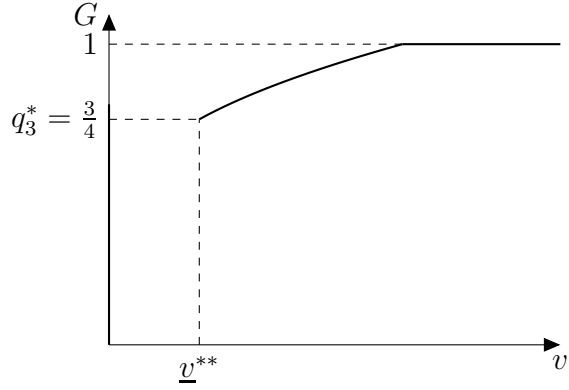
\begin{figure}[h!]
\centering
\begin{tikzpicture}[>=triangle 45, xscale=1.5,yscale=4]
\draw[->] (0,0) -- (4,0) node[below] {$v$};
\draw[->] (0,0) -- (0,1.1) node[left] {${G}$};

\draw[thick, domain=0.8:2.6] 
plot (\x, {1/2+0.1921940017*sqrt(2.819999999*\x-0.563999999)});
\draw[dashed](0.8,0)--(0.8,0.75);
\draw[thick](0,0)--(0,0.8);
\draw[thick] (2.6,1)--(4,1);
\draw[dashed] (0,1)--(2.6,1);
\draw[dashed] (0,0.75)--(0.8,0.75);
\node[left] at (0,1){1};
\node[left] at (0,0.75){$q^*_3=\frac{3}{4}$};
\node[below] at (0.8,0){$\vmin^{**}$};
\end{tikzpicture}

\caption{A graph of a typical cdf from the $G_{\rho}$ family for $n=3$. In the picture, $\rho=\vmin^{**}$, so the depicted cdf is worst-case if $r=c<\vmin^{**}$.}\label{fig:G}
\end{figure}

To proceed, one must first check that $G_{\rho}$ are well-defined. 

\begin{lemma}\label{welldefined}
$G_{\rho}$ is well-defined, i.e. for each $\rho\in[\vmin^{**},m)$, there exists a unique triple $(\lambda_1(\rho),\lambda_2(\rho),q(\rho))$, $\lambda_1(\rho)<0$, $\lambda_2(\rho)>0$, $q(\rho)\in[q^*_n,1)$ such that $G_{\rho}$ has mean $m$ and variance $\sigma^2$. 
\end{lemma}

The proofs of lemma \ref{welldefined} and subsequent lemmata is relegated to the \hyperref[app]{Appendix}. The analysis is enabled by the fact that one can write the mean and variance constraints \eqref{mean2},\eqref{var} as closed-form functions of $\lambda_1,\lambda_2$ and $q$, even though there is generally no closed-form solution for $G_{\rho}$. This is possible since the respective integrals may be rewritten as integrals of the quantile function $G^{-1}(q)$.

Note that if $\vmin^{**}>0$, $G_{\vmin^{**}}(\vmin^{**})\equiv q(\vmin^{**})=q^*_n$.
We now establish the worst-case distribution if $r=c$.

\begin{proposition}\label{wcdistr0}
Suppose $r=c$. Then, the distribution $G_{\max\{\vmin^{**},c\}}$ solves the problem \eqref{controlproblem2}-\eqref{cont2}. 
\end{proposition} 

\pfof{proposition \ref{wcdistr0}}{
By lemma \ref{sufficiency}, it suffices to prove that $G_{\max\{\vmin^{**},c\}}$ minimizes the Lagrangian pointwise. 

Suppose first that $c<\vmin^{**}$ so $\vmin^{**}>0$. Then $G_{\vmin^{**}}(\vmin^{**})=q^*_n$. Take $\lambda_1$ and $\lambda_2$ as coming from the definition of $G_{\vmin^{**}}$ (numbers that make $G_{\vmin^{**}}$ satisfy the mean and variance constraints as equalities). We have $\lambda_1+2\lambda_2\vmin^{**}=n(n-1)z(q^*_n)$. Then it follows from lemma \ref{supplylemma} (with $\lambda_1+2\lambda_2v$ playing the role of $\lambda$) that $G_{\vmin^{**}}$ minimizes the Lagrangian pointwise under the multipliers $\lambda_1$ and $\lambda_2$.

Now suppose $c\geq \vmin^{**}$. Take $\lambda_1$ and $\lambda_2$ as coming from the definition of $G_c$. Then it again follows from lemma \ref{supplylemma} that $G_c$ minimizes the Lagrangian pointwise under the multipliers $\lambda_1$ and $\lambda_2$.\eop}

\subsubsection{Second step}
\begin{figure}[h!]
\centering
\begin{tikzpicture}[>=triangle 45, xscale=4,yscale=4]
\draw[->] (0,0) -- (1.2,0) node[below] {$r$};
\draw[->] (0,0) -- (0,1) node[left] {${R}$};

\draw[thick, domain=0:1] 
plot (\x, {(8/9)*(pow(1-\x,2)*(\x*\x+2)/pow(\x*\x-2*\x+2,2))});
\node[below] at (1,0){$m$};
\node[below] at (0,0){$c$};
\node[circle,fill=red,inner sep=1.5pt,minimum size=0pt] at (0,4/9) {};
\node[left] at (0,4/9){$A$};

\end{tikzpicture}
\begin{tikzpicture}[>=triangle 45, xscale=4,yscale=4]
\draw[->] (0,0) -- (1.2,0) node[below] {$r$};
\draw[->] (0,0) -- (0,1) node[left] {${R}$};

\draw[thick] (0,0.72901)--(0.53874,0.72901);
\draw[dashed] (0.53874,0.72901)--(0.53874,0) node[below] {$\vmin^{**}$};
\node[circle,fill=red,inner sep=1.5pt,minimum size=0pt] at (0,0.72901) {};
\node[left] at (0,0.72901){$A$};
\draw[thick]
(0.53874,0.50173)--
(0.55411,0.48109)--
(0.56949,0.46008)--
(0.58486,0.43874)--
(0.60024,0.4171)--
(0.61561,0.39521)--
(0.63099,0.37313)--
(0.64636,0.35091)--
(0.66174,0.32861)--
(0.67711,0.3063)--
(0.69249,0.28406)--
(0.70787,0.26198)--
(0.72324,0.24013)--
(0.73862,0.21861)--
(0.75399,0.19751)--
(0.76937,0.17693)--
(0.78474,0.15698)--
(0.80012,0.13776)--
(0.81549,0.11937)--
(0.83087,0.10192)--
(0.84624,0.085514)--
(0.86162,0.070263)--
(0.877,0.056265)--
(0.89237,0.043619)--
(0.90775,0.032419)--
(0.92312,0.022754)--
(0.9385,0.014703)--
(0.95387,0.0083425)--
(0.96925,0.0037363)--
(0.98462,0.00094031)--
(1,0);
\node[below] at (1,0){$m$};
\node[below] at (0,0){$c$};
\end{tikzpicture}
\caption{Proof idea. The curves are graphs of $R(\hat{F}_r,r)$. The graph of the worst-case revenue function $\underline{R}(r)$ must lie everywhere weakly below the depicted curve by Step 2 and must pass through the point $A$ by Step 1. $m=\sigma^2=1$; $c=0$; $n=2$ (left), $n=3$ (right). If $n=2$, $\sigma/m=1$ corresponds to ``high'' variance, and $r^*=0$ has to be the unique maxmin price. If $n=3$, $\sigma/m=1$ corresponds to ``low'' variance, and $r^*=0$ might not be the unqiue maxmin price.}\label{fig:proofidea}
\end{figure}
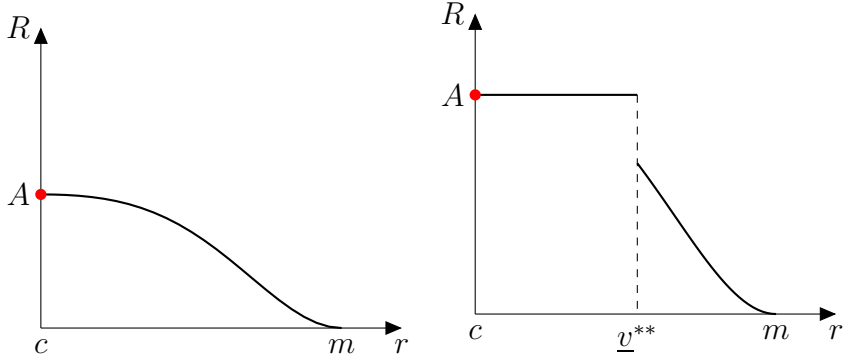

Analogously to section \ref{knownmean}, for every $r$ we construct a feasible distribution $\hat{F}_r$ such that $R(\hat{F}_r,r)\leq \underline{R}(c)$ for all $r$. This implies theorem \ref{main2}.

As in section \ref{knownmean}, the construction of threat distributions $\hat{F}_r$ depends on whether $c>\vmin^{**}$ or $c\leq \vmin^{**}$; equivalently, whether variance $\sigma^2$ is high or low.

\vspace{1em}
\textbf{Case 1 (High Variance).} $c>\vmin^{**}$. 
Define threat distributions $\hat{F}_r$ by 
\[\hat{F}_r:=
\begin{cases}
G_{\vmin^{**}}, & r\in [0,\vmin^{**});\\
G_{r^+}, & r\in [\vmin^{**},c);\\
G_r, & r\in[c,m);\\
\delta_m, & r\geq m,
\end{cases}\]
$G_{\vmin^{**}}$ is the same as the worst-case distribution $F^*$ when $r=c=0$.

Recall that $q(r)\equiv G_r(r)$, the size of the atom of $G_r$ at $r$, and $\lambda_1(r),\lambda_2(r)$ are parameters of $G_r$ (see \eqref{hatFdef}). In the next lemma, we derive closed-form expression for $R(\hat{F}_r,r)$ in terms of $\lambda_1(r)$, $\lambda_2(r), q(r)$ is available, even though there is no closed-form solution for $(\lambda_1(r),\lambda_2(r), q(r))$ themselves. 

\begin{lemma}\label{rexpr}
For $r\geq \vmin^{**}$, 
\beq\label{revenueexpr}
R(G_r,r)=|\lambda_1(r)|(m-q(r)r)-2\lambda_2(r)(m^2+\sigma^2-q(r)r^2)-nz(q(r))rq(r)+c\cdot q^n(r).
\eeq
\end{lemma}

Next, we show that the full derivative $\frac{d R(G_r,r)}{dr}$ also admits a tractable expression. (If $G_r$ were worst-case distributions, this derivative would be computable by a suitable version of envelope theorem; but they are not.)
\begin{lemma}\label{decreasing} 
For $r\geq \vmin^{**}$,
\[\frac{d R(G_r,r)}{dr}=n(n-2)qz-nq^{n-1}q'(r-c).\]
\end{lemma}

Now we are ready to state the key proposition. Recall that $q(r)$ satisfies $\phi(q(r))=1+\frac{\sigma^2}{(m-r)^2}$. Because $\phi'(q)>0$ (shown in the proof of lemma \ref{welldefined}) and  $d(\sigma^2/(m-r)^2)/dr>0$, by implicit function theorem $q(r)$ is a differentiable function with $q'(r)>0$.

\begin{proposition}\label{fhatworks0}
Suppose $c>\vmin^{**}$. Then, $R(\hat{F}_r,r)< \underline{R}(c)$ for all $r\neq c$ and $R(\hat{F}_r,r)=\underline{R}(c)$ for $r=c$.
\end{proposition}

\begin{proof}
The fact that $R(\hat{F}_r,r)=\underline{R}(c)$ for $r=c$ is obvious since $\hat{F}_c$ is a worst-case distribution for $r=c$, as identified by proposition \ref{wcdistr0}.

As $z<0$, $q'>0$, by \eqref{finalderivative2} we have $R'<0$ for $r\geq c$ unless $n=2$ and $r=c$ in which case $R'=0$. Thus, $R(\hat{F}_r,r)$ is strictly decreasing on $[c,m)$ and 
$R(\hat{F}_r,r)< \underline{R}(c)$ for $r\in(c,m)$. 

For $r\geq m$, $R(\hat{F}_r,r)=c$, but $c=\lim\limits_{r\to m^-}  R(\hat{F}_r,r)$. Thus, $c<\underline{R}(c)$.

Consider now $r\in[\vmin^{**},c)$. $R(\hat{F}_r,r)=R(G_{r^+},r)$ differs from $R(G_r,r)$ only by term $(r-c)q^n(r)$ since there is a sale if $v_{(1)}=r^+$ but not if $v_{(1)}=r$. Thus, when $r\in[\vmin^{**},c)$, we get 
\[\frac{dR(\hat{F}_r,r)}{dr}=n(n-2)qz-nq^{n-1}q'\cdot(r-c)+[(r-c)q^n]'=q^{n-1}((n-1)^2q-n(n-2)).\]
For $r> \vmin^{**}$, $q(r)> q^*_n$, and for $r=\vmin^{**}$, $q(r)=q^*_n$. Thus, the above derivative is zero at $r=\vmin^{**}$ and positive at $r\in(\vmin^{**},c)$. Thus, $R(\hat{F}_r,r)$ is strictly increasing on $[\vmin^{**},c)$, which implies that 
$R(\hat{F}_r,r)< \underline{R}(c)$ for $r\in[\vmin^{**},c)$.

Finally, because $R(\hat{F}_r,r)<\underline{R}(c)$ for $r=\vmin^{**}$, this is true for $r<\vmin^{**}$ as well.
\eop
\end{proof}

\vspace{1em}
\textbf{Case 2  (Low Variance).} $c\leq \vmin^{**}$. Define threat distributions $\hat{F}_r$ by 
\[\hat{F}_r:=
\begin{cases}
G_{\vmin^{**}}, & r\in [0,\vmin^{**});\\
G_{r}, & r\in [\vmin^{**},m);\\
\delta_m, & r\geq m,
\end{cases}\]

The proof that $R(\hat{F}_r,r)\leq \underline{R}(c)$ is exactly the same as the proof of proposition \ref{fhatworks2}, with $\vmin^{*}$ replaced by $\vmin^{**}$.

Now, because by the above analysis 
$R(\hat{F}_r,r)\leq \underline{R}(c)$ for all $r$ and because distributions $\hat{F}_r$ are feasible to Nature, we conclude that $\underline{R}(r)\leq R(\hat{F}_r,r)\leq \underline{R}(c)$ for all $r$, which finishes the proof of theorem \ref{main2}.

\subsection{Uniqueness}
Analogously to proposition \ref{unique?}, we establish the following:

\begin{proposition}\label{unique?2}
Suppose the seller knows the mean of value distribution $m$ and an upper bound on variance $\sigma^2$. Then:
\begin{enumerate}
\item (High variance case.) If  $\vmin^{**}<c$, $c$ is the unique maxmin reserve price;
\item (Low variance case.) If $\vmin^{**}\geq c$, all prices $r\in[0,c]$ are maxmin reserve prices.
\end{enumerate}
\end{proposition}

The proof of proposition \ref{unique?2} is identical to the proof of proposition \ref{unique?}. Note that the condition $c\leq \vmin^{**}$ can be alternatively viewed as variance is small enough or the number of bidders is high enough. As in section \ref{knownmean}, part 2 of proposition \ref{unique?2} might weakly explain why reserve price substantially lower than $c$ are observed in practice. Unlike the case in section \ref{knownmean}, in the case of known bound on variance we are not aware of the full characterization of the set of maxmin reserve prices.

\subsection{Example: $n=2$}
In case of two bidders, closed-form expressions for threat distributions $\hat{F}_r$ are available. As $z(q)=q-1$ for $n=2$, by \eqref{hatFdef}, $G^{-1}(q)$ is a linear function wherever it is defined and thus distributions $G_r$ are mixtures of an atom at $r$ with a uniform distribution. The same applies to the worst-case distribution $F^*$ when $r=c$. 

For $n=2$, $\vmin^{**}=\max\{m-\sqrt{3}\sigma,0\}$. $m-\sqrt{3}\sigma$ is simply the lowest point in the support of a uniform distribution with mean $m$ and variance $\sigma^2$. 

Suppose $m-\sqrt{3}\sigma\geq 0$, so $\vmin^{**}=m-\sqrt{3}\sigma$. There are two cases: $c\leq m-\sqrt{3}\sigma$ and $c>m-\sqrt{3}\sigma$. In the first case, which corresponds to either low seller's valuation or low variance, the worst-case distribution for $r=c$ is simply uniform on $[m-\sqrt{3}\sigma,m+\sqrt{3}\sigma]$ and the maxmin revenue is 
\beq\label{r2}
\underline{R}^*_2=m-\frac{\sqrt{3}}{3}\sigma.
\eeq  

Nature's threats work as follows. For all $r\in[0,m-\sqrt{3}\sigma]$, Nature can use this uniform distribution and induce the same revenue (and this distribution is still worst-case for $r\in[0,c]$ as shown in the proof of proposition \ref{unique?2}). For $r\in(m-\sqrt{3}\sigma,m)$, Nature may use a distribution $G_r$ that has an atom of 
\[q(r)=\frac{\sigma^2-(m-r)^2/3}{\sigma^2+(m-r)^2}\]
on $r$ and is uniform on $(r,b(r)]$ where 
\[b(r)=
\frac{1}{2}(3m-r)+\frac{3}{2}\frac{\sigma^2}{m-r}.\]
Note that $b(r)$ grows without bound when $r\to m$; the fact that there is no upper bound on values is important. For $r\geq m$, Nature puts all mass on $m$. 

If $c>m-\sqrt{3}\sigma$ (either high seller's valuation or high variance), the worst-case distribution for $r=c$ is $G_c$ itself. For $r\in[0,m-\sqrt{3}\sigma]$, Nature may use uniform distribution on $[m-\sqrt{3}\sigma,m+\sqrt{3}\sigma]$, for $r\in(m-\sqrt{3}\sigma,c)$ it may use $G_{r^+}$, while for $r\in(c,m)$ it again may use $G_r$.

When $n=2$, we can illustrate the fact that threat distributions $\hat{F}_r$ are generally \emph{not} worst-case distributions. For instance, suppose $c=0$, $m=\sigma^2=1$, $n=2$ (Figure \ref{fig:proofidea}, left) and $r=0.5$. Then $\hat{F}_r$ is is such that $v_i$ is distributed uniformly on $[0.5,4.25]$ with prob. $\frac{4}{15}$ and equal to $0.5$ with prob. $\frac{11}{15}$. Expected revenue under $\hat{F}_r$ is 0.32. However, if Nature uses a distribution which is a mixture of $\delta_{0.5}$ and uniform distribution on $[2.5,b]$ (where $b$ and the size of the atom are pinned down by moments constraints), the expected revenue is approximately $0.2767<0.32$. Numerically, all worst-case distributions have similar gaps in support and are intractable analytically even for $n=2$. 

\subsection{A formula for maxmin revenue}
It may be shown that in the low variance case the formula for maxmin revenue has the same simple form as in \eqref{r2}. Indeed, in the low variance case all prices in $[0,c]$ are maxmin, and the maxmin revenue is equal to that under the price $\vmin^{**}$, as if there were sale when $v_{(1)}=r$. Thus, replacing $c$ with $r$, plugging $r=\vmin^{**}$ in \eqref{revenueexpr}, and then getting rid of $\lambda_1$ and $\lambda_2$ using \eqref{lambda1cond}-\eqref{eq2},  one gets that 
\beq
\underline{R}^*_n=m-\gamma_n\sigma,
\eeq
where $\gamma_n$ depends only on $n$ (in the proof of proposition \ref{asymptoptimal} we give a formula for $\gamma_n$ in terms of $\phi(q^*_n)$).  Thus, the worst-case revenue is simply the mean minus a penalty linear in the standard deviation. However, the values of the penalties $\gamma_n$ are rather unexpected. For instance \footnote{Taking $\sigma/m\leq \sqrt{4.7}$ is enough for variance to be ``low'' if $c=0$.}, 
\[\underline{R}^*_3=m-\frac{\sqrt{470}}{80}\sigma\approx m-0.271\sigma\]
\[\underline{R}^*_4=m-\frac{\sqrt{21604695}}{25515}\sigma\approx m-0.182\sigma\]
\[\underline{R}^*_5=m-\frac{\sqrt{8995616791}}{688128}\sigma\approx m-0.138\sigma.\]

Note that the maxmin revenue is strictly decreasing in variance. This is expected for $n=2$ when the second order statistic is the minimal value, whose expectation is naturally below $m$ and so higher variance reduces it. For higher $n$, the expectation of the second order statistic can be well above $m$ and so larger variance may naturally increase rather than decrease it. The resolution to this paradox is that Nature chooses highly skewed distributions such that $\mathbb{E}v_{(2)}$ is below $\mathbb{E}v_i=m$ for any $n$. We expect the solution to be substantially different, with variance constraint not binding,  when Nature is allowed to choose only from symmetric distributions for $n\geq 3$.

\section{Large number of bidders}\label{largenumber}
\subsection{Asymptotically optimal mechanism}
Throughout the paper so far, we have considered only the issue of optimal reserve price but not the issue of optimal mechanism. This more general question seems to be challenging. One reason for that is that when the distribution $F$ is unknown, one cannot assess whether a given direct mechanism is Bayesian-incentive compatible or not. (However, one can still ask what is a worst-case $F$ for a given mechanism, having in mind the effect of $F$ on the potentially non-truthful equilibrium strategies.)
But even if one assumes dominant-strategy incentive-compatibility of a direct mechanism a priori (a notion independent of $F$), Nature's optimization problem, as seen as a problem of choosing a joint distribution of values, is not a convex problem due to the stochastic independence constraint (a mixture of two product distributions is not in general a product distributions). This precludes the use of strong duality -- a simplification trick used by, e.g., \cite{suzdal2020anoptimal}.

However, it is possible to establish that the second-price auction with a maxmin reserve price (e.g. $r^*=c$) is an asymptotically maxmin mechanism among all ex post individually rational mechanisms when the number of bidders is large.

The argument rests only on the fact that by ex post individual rationality, the revenue of a mechanism is not higher than the social surplus, $\max_i v_i$. In fact, we do not even need to assume that the bidders play a Bayesian equilibrium. All we need is that a mechanism and solution concept are such that the bidders always get a nonnegative payoff. 

Formally, let a mechanism be a tuple $M=(S_1,\ldots,S_n, \textbf{x}(s),\textbf{t}(s))$ where, as usual, $S_i$ is the set of strategies of bidder $i$, and the functions $\textbf{x}(s)$, $\textbf{t}(s)$, with the usual codomains, specify a (possibly randomized) allocation and a vector of transfers for each strategy profile $s$. Let $\mathcal{M}$ be the set of all mechanisms. Let measurable \emph{outcome functions} $x(v),t(v)$ map a vector of values $v$ to an allocation and a vector of transfers. Let $O$ be the set of all outcome functions. A \emph{solution concept} is a correspondence $SC: \mathcal{M}\times \Delta\Rightarrow O$. That is, a solution concept maps a mechanism $M\in\mathcal{M}$ and a value distribution $F\in\Delta$ to a set of outcome functions deemed possible. This set can depend on the distribution, as for the Bayesian equilibrium solution concept. There can also be multiple ``equilibrium'' outcome functions. Given a correspondence $SC$, the seller restricts attention only to mechanisms $M$ such that $SC(M,F)\neq\emptyset$ for all $F\in\Delta$. 

We say that a mechanism $M_0$ is \emph{robustly ex post individually rational} under a solution concept $SC$ and set of distributions $\Delta$ iff, for all $F\in\Delta$, all outcome functions $(x(\cdot),t(\cdot))\in SC(M_0,F)$, all vectors of values $v$, and all $i$, $v_ix_i(v)-t_i(v)\geq 0$. 
The qualifier ``robustly'' refers to the fact the inequality holds for the outcome functions, that may depend on $F$, regardless of $F$. 

For example, the first-price auction is robustly ex post IR under Bayesian equilibrium as winning bidders never pay more than their values, and losers pay nothing. The English auction is robustly ex post IR under the weak solution concept used by \cite{haile2003inference} who only assume that bidders never bid more than their values and never allow an opponent to win at a price they are willing to beat (and under stronger concepts as well). The schemes in \cite{segal2003optimal} are robustly ex post IR under undominated strategies. 

Now, let $R_n(M,F,o)$ be the expected revenue with $n$ bidders under mechanism $M$, distribution $F$ and outcome functions $o\in SC(M,F)$. Let $\underline{R}_n(M)$ be the $n$-bidder revenue guarantee of a mechanism with unknown distribution $F$ and ``equilibrium'' outcome $o\in SC(M,F)$, i.e.
\[\underline{R}(M):=\inf\limits_{F\in\Delta}\inf\limits_{o\in SC(M,F)} R(M,F,o).\]
The worst case over ``equilibrium'' outcomes is taken because if the set of ``equilibria'' depends on $F$ and $F$ is unknown, the seller cannot suggest an equilibrium to play. 

Denote by $SPA(r^*)$ the second-price auction with a maxmin reserve price and let $(x_{DS}(v), t_{DS}(v))$ be the usual outcome functions if bidders play dominant strategies (bid their values) in the SPA. 

Then, a simple SPA with a maxmin deterministic reserve price (e.g., $r^*=c$) is  also an asymptotically maxmin mechanism among all robustly ex post individually rational mechanisms provided dominant strategies are played in the SPA.

\begin{proposition}\label{asymptoptimal}
Suppose $\Delta=\Delta(m,\vmax)$ or $\Delta=\Delta(m,\sigma^2)$. Suppose the solution concept $SC_0$ is such that $SC_0(SPA(r^*),F)=\{(x_{DS}(v), t_{DS}(v))\}$ for all $F\in\Delta$. Then, for any mechanism $M_0$ that is robustly ex post individually rational under $SC_0$ and $\Delta$ and for any $\varepsilon>0$, there exists $N$ such that for all $n>N$, $\underline{R}_n(SPA(r^*))>\underline{R}_n(M_0)-\varepsilon$.
\end{proposition}

\begin{proof}
The revenue of any robustly ex post individually rational mechanism $M$ is not more than the social surplus, $\sum_i t_i\leq \max_i v_i$; thus, the worst-case expected revenue is not more than the worst-case expected surplus. However, the latter is not more than $m$ because Nature can always choose $F=\delta_m$. (This argument is similar to the one used in \cite{koccyiugit2020distributionally} and \cite{he2020correlation}.) Thus, $\underline{R}_n(M)\leq m$.

It remains to show than $\underline{R}_n(SPA(r^*))$ converges to $m$ as $n\to\infty$. For $\Delta=\Delta(m,\vmax)$, we have for all sufficiently high $n$
\[\underline{R}^*_n=m-\alpha_n(\vmax-m),\]
where 
\[\alpha_n=\frac{n}{n-1}\left(1-\frac{1}{(n-1)^2}\right)^{n-2}-1.\]
The fact that $\alpha_n$ converges to zero stems from the fact that 
$\left(1-\frac{1}{(n-1)^2}\right)^{n-2}\sim \exp(-1/n)$ as $n\to\infty$. 

For $\Delta=\Delta(m,\sigma^2)$, we have for all sufficiently high $n$
\[\underline{R}^*_n=m-\gamma_n\sigma,\]
where it follows from \eqref{revenueexpr} and \eqref{lambda1cond} (in the Appendix) that 
\beq\label{gamma}
\gamma_n=\alpha_n\sqrt{(n-1)^2\psi(q^*_n)-1},
\eeq
where $\psi(q)=\phi(q)(1-q)$ (see the proof of lemma \ref{welldefined}).
As $\psi$ is bounded, to prove that $\gamma_n\to 0$ one has to prove that $n\alpha_n\to 0$ as $n\to \infty$. In fact, one may show that
\beq\label{alpharate}
\lim\limits_{n\to\infty}n^2\alpha_n=\frac{1}{2}.
\eeq
Indeed, 
\[n^2\alpha_n\sim n^2\frac{n(\exp(-1/n)-1)+1}{n-1}.\]
Then the result follows from a second-order expansion of $\exp(-1/n)$.
\eop
\end{proof}

Proposition \ref{asymptoptimal} implies that neither a randomized reserve price nor a ``bootstrap'' auction in which each bidder faces individual reserve computed based on an estimate of $F$ inferred from other bidders' reported values (a family of mechanisms investigated by \cite{segal2003optimal}) is significantly better than a simple second-price auction when the seller is concerned about the worst-case performance of a mechanism, possesses only minimal statistical information and the number of bidders is sufficiently large. As discussed in the introduction, Segal's schemes approximate the full-distributional-information revenue well, but this is a criterion different from the worst-case performance.

Note that if the seller knows more than the first two moments, say, three moments of the distribution, then the bound $\inf_F \mathbb{E}_{F\sim\cdots\sim F}\max_i v_i\leq m$ may not hold, i.e. the worst-case expected social surplus may be strictly higher than $m$. In other words, this proof technique fails. It is an open question of whether an SPA with a deterministic maxmin reserve price is an asymptotically maxmin mechanism when more moments are known. 

Another limitation of the above result is that it effectively rules out directly eliciting information about $F$ from the bidders (if they have it). More precisely, the setting above does allow asking about $F$ (strategy sets $S_i$ can be any) but the set $\Delta$ of distributions feasible to Nature does not depend on bidders' reports. An interesting direction of further research is to consider a model with a ``rich'' type space (i.e., types encoding both payoff and belief information) and a partially known payoff type distribution such that Nature can choose only a payoff type distribution satisfying some prior constraints and also being ``close'' to some types' elicitable beliefs. See \cite{luz2013surplus} and \cite{chen2018revisiting} for revenue maximization with rich type spaces.

\subsection{Comparing rates of convergence}
How does the seller's maxmin revenue compare in the two settings we have considered in this paper? The revenue obviously depends on the particular values of the upper bounds. However, noting that in both cases revenue converges to the mean $m$ as the number of bidders tends to infinity we can still get a meaningful comparison by comparing the rates of convergence. This will compare the ``strength'' of Nature in two cases. We also bring into the picture the case of correlated private values, studied by \cite{koccyiugit2020distributionally} and \cite{suzdal2020anoptimal} (maxmin reserve price is again $c$ in this case for all $n$ sufficiently large).

According to results in \cite{koccyiugit2020distributionally} and \cite{suzdal2020anoptimal}, when mean $m$ and upper bound on values $\vmax$ are known, and values can be arbitrarily correlated,
\beq\label{revcorr}
\underline{R}_n(r^*)=m-\frac{\vmax-m}{n-1}
\eeq
for all $n$ sufficiently large.

Thus, from \eqref{gamma}, \eqref{alpharate} and \eqref{revcorr} we make the following 

\textbf{Observation.}
\begin{itemize}
\item When mean $m$ and upper bound on values $\vmax$ are known, and values are iid,
\[m-\underline{R}^*_n=\Theta\left(\frac{1}{n^2}\right)\] as $n\to\infty$;
\item When mean $m$ and upper bound on variance $\sigma^2$ are known and values are iid,
\[m-\underline{R}^*_n=\Theta\left(\frac{1}{n}\right)\] as $n\to\infty$;
\item When mean $m$ and upper bound on values $\vmax$ are known, and values can be arbitrarily correlated,
\[m-\underline{R}^*_n=\Theta\left(\frac{1}{n}\right)\] as $n\to\infty$.
\end{itemize}

Thus, starting with the first setting in which mean $m$ and upper bound on values $\vmax$ are known and values are iid, replacing an upper bound on values with an upper bound on variance has a similar (adverse) effect on revenue as allowing arbitrary correlation in values.  
This suggests that an upper bound on values is a much more stringent constraint on Nature than an upper bound on variance. This is perhaps not surprising in hindsight: a bound on values implies a bound on variance, but not vice versa.

\section{Conclusion}\label{disc}
In this paper, we showed that a seller who (1) faces multiple bidders and conducts a second-price auction with a reserve price; (2) possesses only basic information about value distribution; (3) employs a worst-case perspective can do no better than to set the reserve price to her own valuation. This result adds to an emerging theme in the current robustness literature: an auctioneer's worst-case perspective is associated with low reserve prices.  The result may also help explain empirical observations of low reserve prices.

One may think of several extensions of the present results.
\begin{itemize}
\item \textbf{First-price auctions.} We presented an analysis for a second-price auction. A problem with extending results for first-price auctions by revenue-equivalence is that to ensure that a classic pure-strategy Bayesian equilibrium exists, one must constrain Nature to use atomless distributions only, but the worst-case and threat distributions we used throughout the paper do contain atoms. However, the analysis still goes through as these special distributions can be approximated by a sequence of continuous distributions such that the value of revenue is the same in the limit. For example, it can be done by replacing each atom with uniform distributions on an interval below it and then tweaking the overall distribution a little in such a way that Nature's  constraints hold. Thus, in a first-price auction the seller can still do no better than to set a reserve equal to her own valuation.

\item \textbf{Exactly known variance.} We stated the results in section \ref{knownmeanandvariance} for the case of a known upper bound on variance. However, if the variance is known exactly, the results still hold. First, even though if variance is known exactly the set of feasible distributions is not compact (as noted by \cite{carrasco2018optimal}), the infimum in Nature's problem is well-defined. Second, for $r<m$ all worst-case and threat distributions we consider are such that the variance constraint binds. Third, for $r\geq m$ instead of $m$ one can consider a sequence of distributions $G_{m-1/k}$ defined in section \ref{knownmeanandvariance}. Along this sequence, the probability of no sale converges to one and moments constraints are satisfied with equality. As for the asymptotic result (proposition \ref{asymptoptimal}), the key inequality necessary for the proof, $\inf_F \mathbb{E}_{F\sim\cdots\sim F}\max_i v_i\leq m$, still holds, as under the sequence of binary distributions $F_k$ with the lowest point in support $m-1/k$ and a required mean $m$ and variance $\sigma^2$, $\mathbb{E}_{F_k\sim\cdots\sim F_k}\max_i v_i$ converges to $m$.

\item \textbf{Higher moments.} When the seller knows $F$ exactly (i.e. an infinite number of moments of $F$), any optimal reserve price is strictly above $c$. We showed that when the seller knows up to two moments, an optimal price is $c$. It is interesting to ask what is the minimal number $\underline{K}$ such that set of optimal prices is bounded away from $c$ if the seller knows no fewer than $\underline{K}$ moments. Is it true that $\underline{K}=3$? Does the asymptotic optimality result in section \ref{largenumber} carry over to the case where the seller has more information?

\item \textbf{Eliciting information about $F$ from the bidders.} One may extend the model to ``rich'' type spaces and constrain Nature to choose a distribution ``close'' to some types' beliefs, as discussed in section \ref{largenumber}.

\end{itemize}

\bibliographystyle{te}

\bibliography{bibliography2}
 
\section*{Appendix}\label{app}
\pfof{lemma \ref{welldefined}}{
Fortunately, one can write constraints \eqref{mean2},\eqref{var} as closed-form functions of $\lambda_1,\lambda_2$ and $q$, even though there is no closed-form solution for $G_{\rho}$. This is possible since the respective integrals may be rewritten as integrals of the quantile function $G^{-1}(q)$. Indeed, one gets the following system of three equations with three unknowns $(\lambda_1,\lambda_2,q)$:

\[q(\rho)\rho+\int_{q(\rho)}^1\frac{n(n-1)z(q)-\lambda_1(\rho)}{2\lambda_2(\rho)}dq=m\]
\[q(r)\rho^2+\int_{q(\rho)}^1\left[\frac{n(n-1)z(q)-\lambda_1(\rho)}{2\lambda_2(\rho)}\right]^2dq=m^2+\sigma^2.\]
\beq\label{lambda1cond}
\lambda_1(\rho)+2\lambda_2(\rho)\rho=n(n-1)z(q(\rho)).%
\eeq
(The first two equations stem from the mean and variance constraints \eqref{mean2},\eqref{var} while the third is \eqref{hatFdef}, written for $q=q(\rho)$.) Getting rid of $\lambda_1(\rho)$, one simplifies the first two equations to 
\beq\label{eq1}
n(n-1)\int_{q(\rho)}^1(z(q)-z(q(\rho)))dq=2\lambda_2(\rho)(m-\rho)
\eeq
\beq\label{eq2}
(n(n-1))^2\int_{q(\rho)}^1(z(q)-z(q(\rho)))^2dq=(2\lambda_2(\rho))^2((m-\rho)^2+\sigma^2).
\eeq
Equations \eqref{eq1}-\eqref{eq2} may be collapsed to 
\beq\label{eqq}
\phi(q(\rho))=1+\frac{\sigma^2}{(m-\rho)^2},
\eeq
where $\phi(q)$ is as defined in \eqref{phidef}. We now show that for every $\rho\in[\vmin^{**},m)$ equation \eqref{eqq} admits a unique solution $q(\rho)\in[q^*_n,1)$. To this end, we prove that $\phi'(q)>0$ for $q\in[q^*_n,1)$ and that $\lim\limits_{q\to 1}\phi(q)=+\infty$. 

First, by direct computation, $\phi'(q)$ is proportional to
\[z'(q)\cdot\left(\frac{\int_q^1(z(y)-z(q))^2dy}{1-q}-\left(
\frac{\int_q^1(z(y)-z(q))dy}{1-q}\right)^2\right).\]
$z'(q)>0$ for $q\geq q^*_n$ and the expression in parentheses is positive because it equals the variance of random variable $z(Y)-z(q)$ where random variable $Y$ is distributed uniformly on $[q,1]$. Thus, $\phi'(q)>0$. 

Now, $\phi(q)=\psi(q)/(1-q)$ where $\psi(q)$ is the ratio of mean of square to the square of mean of random variable $z(Y)-z(q)$. Thus, $\psi(q)>1$ for all $q<1$, so $\lim\limits_{q\to 1}\phi(q)=+\infty$. 

It follows from the definition of $\vmin^{**}$ that for every $\rho\in[v^{**},m)$, $\phi(q^*_n)\leq 1+\frac{\sigma^2}{(m-\rho)^2}$; we also have  $\lim\limits_{q\to 1}\phi(q)=+\infty>1+\frac{\sigma^2}{(m-\rho)^2}$. Because $\phi(q)$ is continuous and strictly increasing, there exists a unique $q(\rho)\in [q^*_n,1)$ solving \eqref{eqq}. 

The uniqueness and signs of $\lambda_1(\rho)$, $\lambda_2(\rho)$ follow from \eqref{eq1} and \eqref{lambda1cond}.\eop}

\pfof{lemma \ref{rexpr}}{
Note that $G_r$ has an atom at $r$ and a density $g_r$ for $v>r$. Because there is no sale when $v_{(1)}=r$, 
\begin{multline}
R(G_r,r)=\mathbb{E}(v_{(2)}1_{\{v_{(2)}>r\}})+r\cdot n(1-q(r))q^{n-1}(r)+cq^n(r)=\\
=\int_r^{\infty}
v\cdot n(n-1)(G_r^{n-2}(v)-G^{n-1}_r(v))g_r(v)dv-nz(q(r))rq(r)+cq^n(r).
\end{multline}
From \eqref{hatFdef}, one deduces that $g_r(v)=\frac{2\lambda_2(r)}{n(n-1)z'(G_r(v))}$. Thus, 
\[\mathbb{E}(v_{(2)}1_{\{v_{(2)}>r\}})=
\int_r^{\infty}2\lambda_2(r)v\frac{-z(G_r(v))}{z'(G_r(v))}dv.\]
Plugging $v=G^{-1}_r(q)$ and using \eqref{hatFdef} again, one gets
\[\mathbb{E}(v_{(2)}1_{\{v_{(2)}>r\}})=
-\int_{q(r)}^1n(n-1)z(q)\frac{n(n-1)z(q)-\lambda_1(r)}{2\lambda_2(r)}dq.\]
Mean and variance constraints read as 
\[q(r)r+\int_{q(r)}^1\frac{n(n-1)z(q)-\lambda_1(r)}{2\lambda_2(r)}dq=m\]
\[q(r)r^2+\int_{q(r)}^1\left[\frac{n(n-1)z(q)-\lambda_1(r)}{2\lambda_2(r)}\right]^2dq=m^2+\sigma^2.\]

Then, the fact that $\mathbb{E}(v_{(2)}1_{\{v_{(2)}>r\}})=-\lambda_1(r)(m-q(r)r)-2\lambda_2(r)(m^2+\sigma^2-q(r)r^2)$ stems directly from the last three equations. The formula for revenue then follows.
\eop
}

\pfof{lemma \ref{decreasing}}{
Now consider the derivative of  $\tilde{R}(r):=\mathbb{E}(v_{(2)}1_{\{v_{(2)}>r\}})=-\lambda_1(r)(m-q(r)r)-2\lambda_2(r)(m^2+\sigma^2-q(r)r^2)$. We omit arguments of functions for brevity. From \eqref{lambda1cond}, we get that
\[\tilde{R}=2\lambda_2r(r-m)-2\lambda_2((m-r)^2+\sigma^2)-n(n-1)zm+n(n-1)qrz.\]
\[\tilde{R}'=[2\lambda_2(r-m)]'r+2\lambda_2(m-r)-2\lambda_2'((m-r)^2+\sigma^2)-n(n-1)z'q'm+n(n-1)(qrz)'.\]
Using both \eqref{eq1} and the differentiated version of \eqref{eq1}, one gets
\[\tilde{R}'=n(n-1)((qr)'z-(m-r)z'q')+n(n-1)\int_q^1(z(x)-z(q))dx-2\lambda_2'((m-r)^2+\sigma^2).\]
However, differentiating \eqref{eq2} one gets
\[\lambda_2'((m-r)^2+\sigma^2)=-n^2(n-1)^2\frac{q'z'\int_q^1(z(x)-z(q))dx}{4\lambda_2}+\lambda_2(m-r),\]
which, using \eqref{eq1} again, is equivalent to 
\[2\lambda_2'((m-r)^2+\sigma^2)=n(n-1)\left[-q'z'(m-r)+\int_q^1(z(x)-z(q))dx\right].\]
Thus, 
\[\tilde{R}'=n(n-1)(qr)'z.\]
Because $R=\tilde{R}-nzqr+cq^n$, we  get
\[R'=n(n-1)(qr)'z-nqrz'q'-n(qr)'z=n(n-2)(qr)'z-nqrz'q'+nq^{n-1}q'c.\]
Since $qz'=(n-2)z+q^{n-1}$, after simplifications we finally get
\beq\label{finalderivative2}
R'=n(n-2)qz-nq^{n-1}q'(r-c).
\eeq

\eop}

\end{document}